\documentclass[11pt,leqno]{article}

\usepackage{amsmath}
\usepackage{amsthm}
\usepackage{amssymb}

\theoremstyle{plain}
\newtheorem{theorem}{Theorem}[section]
\newtheorem{lemma}[theorem]{Lemma}
\newtheorem{proposition}[theorem]{Proposition}

\theoremstyle{definition}
\newtheorem{definition}[theorem]{Definition}

\theoremstyle{remark}
\newtheorem{remark}[theorem]{Remark}

\numberwithin{equation}{section}

\begin{document}

\title{\textbf{A mathematical proof that\\
the transition to a superconducting state is\\
a second-order phase transition}}

\author{Shuji Watanabe\\
Division of Mathematical Sciences\\
Graduate School of Engineering, Gunma University\\
4-2 Aramaki-machi, Maebashi 371-8510, Japan\\
Email: watanabe@fs.aramaki.gunma-u.ac.jp}

\date{}

\maketitle

\begin{abstract}
We deal with the gap function and the thermodynamical potential in the BCS-Bogoliubov theory of superconductivity, where the gap function is a function of the temperature $T$ only. We show that the squared gap function is of class $C^2$ on the closed interval $[\,0,\,T_c\,]$ and point out some more properties of the gap function. Here, $T_c$ stands for the transition temperature. On the basis of this study we then give, examining the thermodynamical potential, a mathematical proof that the transition to a superconducting state is a second-order phase transition. Furthermore, we obtain a new and more precise form of the gap in the specific heat at constant volume from a mathematical point of view.

\medskip


\medskip

\end{abstract}


\section{Introduction}

Let $\varepsilon>0$ be small enough and let us fix it unless otherwise stated. Let $k_B>0$ and $\omega_D>0$ stand for the Boltzmann constant and for the Debye frequency, respectively. We denote Planck's constant by $h \; (>0)$ and set $\hslash=h/(2\pi)$. Let $\mu>0$ stand for the chemical potential. Let $N(\xi)\geq 0$ stand for the density of states per unit energy at the energy $\xi$ \quad $(-\mu\leq \xi<\infty)$ and let $N_0=N(0)>0$. Here, $N_0$ stands for the density of states per unit energy at the Fermi surface $(\xi=0)$. Let $U_0>0$ be a constant. 

It is well known that superconductivity occurs at temperatures below the temperature $T_c>0$ called the transition temperature. We now define it.
\begin{definition}\label{dfn:temperature}
The transition temperature is the temperature $T_c>0$ satisfying
\[
\frac{1}{\, U_0N_0\,}=\int_{\displaystyle{\varepsilon}}^{\displaystyle{\hslash\omega_D/(2k_BT_c)}} \frac{\,\tanh \eta\,}{\eta}\,d\eta\,.
\]
\end{definition}

Generally speaking, the gap function is a function both of the temperature $T$ and of wave vector. In this paper we however regard the gap function as a function of the temperature $T$ only, and denoted it by $\Delta(T)$ $(\geq 0)$. Such a situation is considered in the BCS-Bogoliubov theory \cite{bcs, bogoliubov}, and is accepted widely in condensed matter physics (see e.g. \cite[(7.118), p.~250]{niwa}, \cite[(11.45), p.~392]{ziman}). See also
\cite{watanabe} and \cite{watanabe2} for related material. The gap function satisfies the following nonlinear integral equation called the gap equation (c.f. \cite{bcs}):
\   For $0\leq T\leq T_c$,
\begin{equation}\label{eq:gapequation}
1=U_0N_0
\int_{\displaystyle{2k_BT_c\,\varepsilon}}^{\displaystyle{\hslash\omega_D}}\frac{1}{\,\sqrt{\,\xi^2+f(T)\,}\,}
\tanh\frac{\, \sqrt{\,\xi^2+f(T)\,}\,}{2k_BT}\,d\xi.
\end{equation}
Here, for later convenience, the squared gap function is denoted by $f$, i.e., $f(T)=\Delta(T)^2$.

\begin{remark}
We introduce the cutoff $\varepsilon$ in Definition \ref{dfn:temperature} and in the gap equation (\ref{eq:gapequation}). When $\varepsilon=0$, Definition \ref{dfn:temperature} and the gap equation (\ref{eq:gapequation}) reduce to those in the BCS-Bogoliubov theory \cite{bcs, bogoliubov}. Furthermore, when $\varepsilon=0$, the thermodynamical potential $\Omega$ in Definition \ref{dfn:thpo} below reduces to that in the BCS-Bogoliubov theory (see also (\ref{eq:omegan}) and (\ref{eq:delta}) below). See e.g. Niwa \cite[sec.~7.7.3, p.~255]{niwa}.
\end{remark}

The gap equation (\ref{eq:gapequation}) is a simplified one, and the gap equation with a more general potential is studied extensively. Odeh \cite{odeh} and Billard and Fano \cite{billardfano} established the existence and uniqueness of the positive solution to the gap equation with a more general potential in the case $T=0$. In the case $T\geq 0$, Vansevenant \cite{vansevesant} and Yang \cite{yang} determined the transition temperature and showed that there is a unique positive solution to the gap equation with a more general potential. Recently Hainzl, Hamza, Seiringer and Solovej \cite{hhss}, and Hainzl and Seiringer \cite{haizlseiringer} proved that the existence of a positive solution to the gap equation with a more general potential is equivalent to the existence of a negative eigenvalue of a certain linear operator to show the existence of a transition temperature.

Let $f(T)$ be as in (\ref{eq:gapequation}) and set
\begin{eqnarray}\nonumber
\quad \qquad \Omega_S(T)&=&\Omega_N(T)+\delta(T),\\ \label{eq:omegan}
\quad \qquad \Omega_N(T)&=&-2N_0\int_{\displaystyle{2k_BT_c\,\varepsilon}}^{\displaystyle{\hslash\omega_D}} \xi\,d\xi-4N_0k_BT\int_{\displaystyle{2k_BT_c\,\varepsilon}}^{\displaystyle{\hslash\omega_D}}
 \ln\left( 1+e^{\displaystyle{-\xi/(k_BT)}} \right)\,d\xi \\ \nonumber
& &+V(T),\qquad T>0,\\ \label{eq:delta}
\delta(T)&=&\frac{\, f(T)\,}{U_0}-2N_0
\int_{\displaystyle{2k_BT_c\,\varepsilon}}^{\displaystyle{\hslash\omega_D}} \left\{ \sqrt{\xi^2+f(T)}-\xi\right\}\,d\xi \\ \nonumber
& &-4N_0k_BT
\int_{\displaystyle{2k_BT_c\,\varepsilon}}^{\displaystyle{\hslash\omega_D}} \ln \frac{1+e^{-\displaystyle{\sqrt{\xi^2+f(T)}/(k_BT)}}}{1+e^{-\displaystyle{\xi/(k_BT)}}} \,d\xi,\quad 0<T\leq T_c,\\ \label{eq:v}
\,\quad V(T)&=&2
 \int_{\displaystyle{-\mu}}^{\displaystyle{-\hslash\omega_D}}
 \xi\, N(\xi)\,d\xi
 -2k_BT\int_{\displaystyle{-\mu}}^{\displaystyle{-\hslash\omega_D}}
 N(\xi)\ln\left( 1+e^{\displaystyle{\,\xi/(k_BT)}} \right)\,d\xi
 \\ \nonumber
& &-2k_BT\int_{\displaystyle{\hslash\omega_D}}^{\infty} N(\xi)
 \ln\left( 1+e^{\displaystyle{-\xi/(k_BT)}} \right)\,d\xi,\qquad T>0.
\end{eqnarray}

\begin{remark}\label{rmk:nxi}
Since $N(\xi)=O(\sqrt{\xi})$ as $\xi \to \infty$, the integral on the right side of (\ref{eq:v})
\[
\int_{\displaystyle{\hslash\omega_D}}^{\infty} N(\xi)
 \ln\left( 1+e^{\displaystyle{-\xi/(k_BT)}} \right)\,d\xi
\]
is well defined for $T>0$.
\end{remark}

\begin{definition}\label{dfn:thpo}
Let $\Omega_S(T)$ and $\Omega_N(T)$ be as above. The thermodynamical potential $\Omega$ is defined by
\[
\Omega(T)=\left\{ \begin{array}{ll}\displaystyle{
 \Omega_S(T)} \qquad &(0<T\leq T_c),\\
\noalign{\vskip0.3cm} \displaystyle{
 \Omega_N(T)} \qquad &(T>T_c).
\end{array}\right.
\]
\end{definition}

\begin{remark}
Generally speaking, the thermodynamical potential $\Omega$ is a function of the temperature $T$, the chemical potential $\mu$ and the volume of our physical system. Fixing the values of $\mu$ and of the volume of our physical system, we deal with the dependence of $\Omega$ on the temperature $T$ only.
\end{remark}

\begin{remark}
Hainzl, Hamza, Seiringer and Solovej \cite{hhss}, and Hainzl and Seiringer \cite{haizlseiringer} studied the gap equation with a more general potential examining the thermodynamic pressure.
\end{remark}

\begin{definition}
We say that the transition to a superconducting state
at the transition temperature $T_c$ is a second-order phase transition
if the following conditions are fulfilled:

{\rm (a)}\quad The thermodynamical potential $\Omega$, regarded as a function of $T$, is of class $C^1$ on $(0,\,\infty)$.

{\rm (b)}\quad The second-order derivative
$\left( \partial^2\Omega/\partial T^2\right)$ is continuous on
$(0,\,\infty) \setminus \{ T_c\}$ and is discontinuous at $T=T_c$.
\end{definition}

\begin{remark}
Condition (a) implies that the entropy $\displaystyle{S=
-\left( \partial\Omega/\partial T\right)}$ is continuous on
$(0,\,\infty)$ and that, as a result, no latent heat is observed
at $T=T_c$. On the other hand, (b) implies that the specific heat
at constant volume, $\displaystyle{C_V=-T\left( \partial^2\Omega/\partial T^2\right)}$, is discontinuous at $T=T_c$. See Proposition \ref{prp:specificheat} below, which gives a new and more precise form of the gap $\Delta C_V$ in the specific heat at constant volume at $T=T_c$ from a mathematical point of view.
\end{remark}

From a physical point of view, it is pointed out that the transition from a normal state to a superconducting state is a second-order phase transition. But a mathematical proof of this statement has not been given yet as far as we know. In this paper we first show that there is a unique solution: $T \mapsto f(T)$ of class $C^2$ on the closed interval $[\,0,\,T_c\,]$ to the gap equation (\ref{eq:gapequation}) and point out some more properties of the gap function. Examining the thermodynamical potential $\Omega$, we then give a mathematical proof that the transition to a superconducting state at the transition temperature $T_c$ is a second-order phase transition. Furthermore, we obtain a new and more precise form of the gap in the specific heat at constant volume from a mathematical point of view.

The paper proceeds as follows. In section 2 we state our main results without proof. In sections 3 and 4 we study some properties of the function $F$ defined by (\ref{eq:functionF}) below. On the basis of this study, in sections 5 and 6, we prove our main results in a sequence of lemmas.

\section{Main results}

Let
\[
h(T,\,Y,\,\xi)=\left\{ \begin{array}{ll}\displaystyle{
 \frac{1}{\,\sqrt{\,\xi^2+Y\,}\,}
 \tanh \frac{\, \sqrt{\,\xi^2+Y\,}\,}{2k_BT}
 } \quad &(0<T\leq T_c\,,\quad Y\geq 0),\\
\noalign{\vskip0.3cm} \displaystyle{
 \frac{1}{\,\sqrt{\,\xi^2+Y\,}\,}
 } \quad &(T=0, \quad Y>0)
\end{array}\right.
\]
and set
\begin{equation}\label{eq:functionF}
F(T,\,Y)=\int_{2k_BT_c\,\varepsilon}^{\hslash\omega_D}
h(T,\,Y,\,\xi)\,d\xi-\frac{1}{\,U_0N_0\,}\,.
\end{equation}
Set also
\begin{equation}\label{eq:delta0}
\Delta_0=\frac{\hslash\omega_D}{\,\sinh\frac{1}{\,U_0N_0\,}\,}, \quad
\Delta=\frac{\,
\sqrt{\left\{\hslash\omega_D-2k_BT_c\,\varepsilon\, e^{1/(U_0N_0)}\right\}
\left\{\hslash\omega_D-2k_BT_c\,\varepsilon\, e^{-1/(U_0N_0)}\right\}
}\,}{\,\sinh\frac{1}{\,U_0N_0\,}\,}.
\end{equation}
Since $\varepsilon>0$ is small enough, it follows that $\Delta_0>\Delta$.

We consider the function $F$ on the following domain $W\subset \mathbb{R}^2$:
\[
W=W_1\cup W_2\cup W_3\cup W_4\,,
\]
where
\begin{eqnarray}\nonumber
W_1&=&\left\{ (T,\,Y)\in\mathbb{R}^2:\; 0<T<T_c\,,\; 0<Y<2\,\Delta_0^2 \right\},\\ \nonumber
W_2&=&\left\{ (0,\,Y)\in\mathbb{R}^2:\; 0<Y<2\,\Delta_0^2 \right\},\\
 \nonumber
W_3&=&\left\{ (T,\,0)\in\mathbb{R}^2:\; 0<T\leq T_c \right\},\\ \nonumber
W_4&=&\left\{ (T_c\,,\,Y)\in\mathbb{R}^2:\; 0<Y<2\,\Delta_0^2 \right\}.
\end{eqnarray}

\begin{remark}
The gap equation (\ref{eq:gapequation}) is rewritten as $F(T,\,Y)=0$, where $Y$ corresponds to $f(T)$ $\left( =\Delta(T)^2\right)$.
\end{remark}

Let $g$ be given by
\begin{equation}\label{eq:fng}
g(\eta)= \left\{ \begin{array}{ll}\displaystyle{
\frac{1}{\,\eta^2\,}\left( \frac{1}{\,\cosh^2\eta \,}
 -\frac{\,\tanh\eta\,}{\eta}\right) } \qquad &(\eta>0),\\
\noalign{\vskip0.3cm} \displaystyle{
-\frac{\,2\,}{\,3\,} } &(\eta=0).
\end{array}\right.
\end{equation}
Note that $g(\eta)<0$.

Our main results are the following.
\begin{proposition}\label{prp:gapfunction}
Let $F$ be as in (\ref{eq:functionF}) and $\Delta$ as in (\ref{eq:delta0}). Then there is a unique solution: $T \mapsto Y=f(T)$ of class $C^2$ on the closed interval $[\,0,\,T_c\,]$ to the gap equation $F(T,\,Y)=0$ such that the function $f$ satisfies $f(0)=\Delta^2$ and $f(T_c)=0$, and is monotonically decreasing on $[0,\,T_c]$:
\[
f(0)=\Delta^2>f(T_1)>f(T_2)>f(T_c)=0, \qquad 0<T_1<T_2<T_c\,.
\]
Furthermore the value of the derivative $f'$ at $T=T_c$ is given by
\[
f'(T_c)=8\,k_B^2T_c\,\frac{\,\displaystyle{
\int_{\displaystyle{\varepsilon}}^{\displaystyle{\hslash\omega_D/(2k_BT_c)}} \frac{d\eta}{\,\cosh^2\eta\,}
}\,}{\,\displaystyle{
\int_{\displaystyle{\varepsilon}}^{\displaystyle{\hslash\omega_D/(2k_BT_c)}} g(\eta)\,d\eta
}\,}<0\,.
\]
\end{proposition}

\begin{theorem}\label{thm:phasetransition}
The transition to a superconducting state at the transition temperature $T_c$ is a second-order phase transition, and the following relation holds at the transition temperature $T_c$:
\[
\lim_{T\uparrow T_c} \frac{\,\partial^2\Omega\,}{\,\partial T^2\,}(T)-
\lim_{T\downarrow T_c} \frac{\,\partial^2\Omega\,}{\,\partial T^2\,}(T)
=\frac{\,2N_0f'(T_c)\,}{T_c} \left(
\frac{1}{\,1+e^{\displaystyle{2\varepsilon}} \,}-
\frac{1}{\,1+e^{\displaystyle{\hslash\omega_D/(k_BT_c)}} \,} \right),
\]
where $f'(T_c)$ is given by Proposition \ref{prp:gapfunction}.
\end{theorem}

Setting $\varepsilon=0$ in the results of Proposition \ref{prp:gapfunction} and Theorem \ref{thm:phasetransition} immediately yields the following.
\begin{proposition}\label{prp:specificheat}
Let $T_c$ satisfy
\[
\frac{1}{\, U_0N_0\,}=\int_0^{\displaystyle{\hslash\omega_D/(2k_BT_c)}} \frac{\,\tanh \eta\,}{\eta}\,d\eta
\]
and let $f'(T_c)$ be given by
\[
f'(T_c)=8\,k_B^2T_c\,\frac{\,\displaystyle{
\int_0^{\displaystyle{\hslash\omega_D/(2k_BT_c)}}
 \frac{d\eta}{\,\cosh^2\eta\,}
}\,}{\,\displaystyle{
\int_0^{\displaystyle{\hslash\omega_D/(2k_BT_c)}} g(\eta)\,d\eta
}\,}<0\,.
\]
Then the gap $\Delta C_V$ in the specific heat at constant volume, $\displaystyle{C_V=-T\left( \partial^2\Omega/\partial T^2\right)}$, at the transition temperature $T_c$ is given by the form
\begin{equation}\label{eq:specificheat}
\Delta C_V=-N_0f'(T_c)\tanh \frac{\hslash\omega_D}{\,2k_BT_c\,}>0.
\end{equation}
\end{proposition}

\begin{remark}
A form similar to (\ref{eq:specificheat}) has already been obtained by a different method in the context of theoretical condensed matter physics, but it is an approximate one. However the form (\ref{eq:specificheat}) is a more precise one obtained in the context of mathematics.
\end{remark}

\section{The first-order partial derivatives of the function $F$}

In this section we deal with the first-order partial derivatives of the
function $F$ and show that $F$ is of class $C^1$ on $W$.

A straightforward calculation gives the following.
\begin{lemma}\label{lm:gproperty}
Let $g$ be as in (\ref{eq:fng}). Then the function $g$ is of class $C^1$
on $[0,\,\infty)$ and satisfies
\[
g(\eta)<0,\qquad g'(0)=0,\qquad
\lim_{\eta\to\infty}g(\eta)=\lim_{\eta\to\infty}g'(\eta)=0.
\]
\end{lemma}

\begin{lemma}\label{lm:FTFYW}
The partial derivatives $\displaystyle{\frac{\,\partial F\,}{\,\partial T\,}}$ and $\displaystyle{\frac{\,\partial F\,}{\,\partial Y\,}}$ exist on $W$, and are given as follows. At $(T,\,Y)\in W \setminus W_2$ ,
\[
\left\{ \begin{array}{ll}\displaystyle{
 \frac{\,\partial F\,}{\,\partial T\,}(T,\,Y)=-\frac{1}{\,2k_BT^2\,}
 \int_{2k_BT_c\,\varepsilon}^{\hslash\omega_D}
 \frac{d\xi}{\,\cosh^2 \displaystyle{
 \frac{\,\sqrt{\xi^2+Y}\,}{\,2k_BT\,} } \,}\,,
 } \qquad & \, \\
\noalign{\vskip0.3cm} \displaystyle{
 \frac{\,\partial F\,}{\,\partial Y\,}(T,\,Y)=\frac{1}{\,2(2k_BT)^3\,}
 \int_{2k_BT_c\,\varepsilon}^{\hslash\omega_D}
 g\left( \frac{\,\sqrt{\xi^2+Y}\,}{\,2k_BT\,} \right)\, d\xi
 } \qquad & \,
\end{array}\right.
\]
and at $(0,\,Y)\in W_2$ ,
\[
\left\{ \begin{array}{ll}\displaystyle{
 \frac{\,\partial F\,}{\,\partial T\,}(0,\,Y)=0,
 } \qquad & \, \\
\noalign{\vskip0.3cm} \displaystyle{
 \frac{\,\partial F\,}{\,\partial Y\,}(0,\,Y)
 =-\frac{1}{\,2\,}\int_{2k_BT_c\,\varepsilon}^{\hslash\omega_D}
 \frac{d\xi}{\,( \sqrt{\xi^2+Y} )^3\,}\,.
 } \qquad & \,
\end{array}\right.
\]
\end{lemma}

Lemmas \ref{lm:gproperty} and \ref{lm:FTFYW} immediately give the following.
\begin{lemma}\label{lm:FTFYminus}\quad At $(T,\,Y)\in W\setminus W_2$,
\[
\frac{\,\partial F\,}{\,\partial T\,}(T,\,Y)<0,\qquad
\frac{\,\partial F\,}{\,\partial Y\,}(T,\,Y)<0.
\]
\end{lemma}

We now study the continuity of the functions $F$, $(\partial F/\partial T)$ and $(\partial F/\partial Y)$ on $W$.

\begin{lemma}\label{lm:FC1onW1}
The partial derivatives $\displaystyle{\frac{\,\partial F\,}{\,\partial T\,}}$ and $\displaystyle{\frac{\,\partial F\,}{\,\partial Y\,}}$ are continuous on $W_1$. Consequently, the function $F$ is of class $C^1$ on $W_1$.
\end{lemma}

\begin{proof}
It is enough to show that the functions: $(T,\,Y)\mapsto I_1(T,\,Y)$ and $(T,\,Y)\mapsto I_2(T,\,Y)$ (see Lemma \ref{lm:FTFYW}) are continuous at $(T_0,\,Y_0)\in W_1$. Here,
\begin{equation}\label{eq:I1I2}
I_1(T,\,Y)=\int_{2k_BT_c\,\varepsilon}^{\hslash\omega_D}
 \frac{d\xi}{\,\cosh^2\eta\,}\,,\quad
I_2(T,\,Y)=\int_{2k_BT_c\,\varepsilon}^{\hslash\omega_D}
g(\eta)\, d\xi\,,\quad \eta=\frac{\,\sqrt{\,\xi^2+Y\,}\,}{2k_BT}\,.
\end{equation}
Set $\displaystyle{\eta_0=\frac{\,\sqrt{\,\xi^2+Y_0\,}\,}{2k_BT_0}}$.\quad
Since $(T,\,Y)\in W_1$ is close to $(T_0,\,Y_0)\in W_1$, it follows that
$T>T_0/2$. Then
\begin{eqnarray}\nonumber
& &\left| I_1(T,\,Y)-I_1(T_0,\,Y_0)\right| \\ \nonumber
&\leq& \int_{2k_BT_c\,\varepsilon}^{\hslash\omega_D}
 \left|\left( \frac{1}{\,\cosh\eta\,}+\frac{1}{\,\cosh\eta_0\,} \right)
 \frac{\,\cosh\eta-\cosh\eta_0\,}{\,\cosh\eta\,\cosh\eta_0\,} \right|
 \,d\xi \\ \nonumber
&\leq& 2\hslash\omega_D\sinh
 \frac{\,\sqrt{\,\hslash^2\omega_D^2+2\,\Delta_0^2\,}\,}{k_BT_0}
 \left( \frac{\,\sqrt{\,\hslash^2\omega_D^2+2\,\Delta_0^2\,}\,}{k_BT_0^2}|T-T_0|+\frac{|Y-Y_0|}{\,k_BT_0\sqrt{Y_0}\,}\right), \\ \nonumber
& &\left| I_2(T,\,Y)-I_2(T_0,\,Y_0)\right| \\ \nonumber
&\leq& \int_{2k_BT_c\,\varepsilon}^{\hslash\omega_D}
 \left| g(\eta)-g(\eta_0) \right|\,d\xi \\ \nonumber
&\leq& \hslash\omega_D\,\displaystyle{ \max_{\eta\geq 0}|g'(\eta)| }
 \left( \frac{\,\sqrt{\,\hslash^2\omega_D^2+2\,\Delta_0^2\,}\,}{k_BT_0^2}|T-T_0|+\frac{|Y-Y_0|}{\,k_BT_0\sqrt{Y_0}\,}\right).
\end{eqnarray}
Thus the functions: $(T,\,Y)\mapsto I_1(T,\,Y)$ and $(T,\,Y)\mapsto I_2(T,\,Y)$, and hence $(\partial F/\partial T)$ and $(\partial F/\partial Y)$ are continuous at $(T_0,\,Y_0)\in W_1$.
\end{proof}

\begin{lemma}
\quad The function $F$ is continuous on $W$.
\end{lemma}

\begin{proof}
Note that $F$ is continuous on $W_1$ by Lemma \ref{lm:FC1onW1}. We then show that $F$ is continuous on $W_2$.

Let $(0,\,Y_0)\in W_2$ and let $(T,\,Y)\in W_1\cup W_2$. Since $(T,\,Y)$ is close to $(0,\,Y_0)$, it follows that $Y>Y_0/2$. Then, by (\ref{eq:functionF}),
\begin{eqnarray}\nonumber
& &\left| F(T,\,Y)-F(0,\,Y_0)\right| \\ \nonumber
&\leq& \int_{2k_BT_c\,\varepsilon}^{\hslash\omega_D} \left\{
 \frac{\,1-\tanh\frac{\,\sqrt{\,\xi^2+Y\,}\,}{2k_BT}\,}{\sqrt{\,\xi^2+Y_0\,}}+\left| \frac{1}{\,\sqrt{\,\xi^2+Y\,}\,}
 -\frac{1}{\,\sqrt{\,\xi^2+Y_0\,}\, } \right| \right\}\,d\xi
 \\ \nonumber
&\leq& \hslash\omega_D \left\{
 \frac{1}{\,\sqrt{Y_0}\,}
 \left( 1-\tanh\frac{\,\sqrt{\,Y_0/2\,}\,}{2k_BT} \right)
 +\frac{\,2\,|Y-Y_0|\,}{\,(\sqrt{2}+1)Y_0^{3/2}\,} \right\}.
\end{eqnarray}
Thus $F$ is continuous on $W_2$. Similarly we can show the continuity of $F$ on $W_3$, and on $W_4$.
\end{proof}

\begin{lemma}\label{lm:FC1onW}
The partial derivatives $\displaystyle{\frac{\,\partial F\,}{\,\partial T\,}}$ and $\displaystyle{\frac{\,\partial F\,}{\,\partial Y\,}}$ are continuous on $W$. Consequently, the function $F$ is of class $C^1$ on $W$.
\end{lemma}

\begin{proof}
Note that $(\partial F/\partial T)$ and $(\partial F/\partial Y)$ are continuous on $W_1$ by Lemma \ref{lm:FC1onW1}. We then show that $(\partial F/\partial T)$ and $(\partial F/\partial Y)$ are continuous at $(T_c\,,\,0)\in W_3$. We can show the continuity of those functions at other points in $W$ similarly.

\textit{Step 1}. Let $(T,\,Y)\in W_1$. We show
\[
\frac{\,\partial F\,}{\,\partial T\,}(T,\,Y) \to
 \frac{\,\partial F\,}{\,\partial T\,}(T_c\,,\,0),\;\;
\frac{\,\partial F\,}{\,\partial Y\,}(T,\,Y) \to
\frac{\,\partial F\,}{\,\partial Y\,}(T_c\,,\,0) \quad \mbox{as}\;
(T,\,Y) \to (T_c\,,\,0).
\]
Since $(T,\,Y)$ is close to $(T_c\,,\,0)$, it then follows that $T_c/2<T<T_c$. Set $\eta_0=\frac{\,\sqrt{\hslash^2\omega_D^2+2\,\Delta_0^2}\,}{k_BT_c}$. Then
\begin{eqnarray}\nonumber
& & \left| \frac{1}{\,T^2\cosh^2 \frac{\,\sqrt{\xi^2+Y}\,}{\,2k_BT\,}\,}
-\frac{1}{\,T_c^2\cosh^2 \frac{\,\xi\,}{\,2k_BT_c\,} \,} \right| \\ \nonumber
&\leq& \frac{\,8\cosh \eta_0\,}{T_c^3} \left\{ \left| T-T_c \right|
 \left( \cosh \eta_0+\eta_0 \sinh \eta_0 \right)
 +\frac{\,\sqrt{Y}\,}{4k_B} \sinh \eta_0 \right\},
\end{eqnarray}
and hence \quad $\displaystyle{
(\partial F/\partial T)(T,\,Y)-(\partial F/\partial T)(T_c\,,\,0) \to 0 \quad \mbox{as}\quad (T,\,Y) \to (T_c\,,\,0)}$.

Since
\[
\left| g\left( \frac{\,\sqrt{\,\xi^2+Y\,}\,}{2k_BT} \right)
 -g\left( \frac{\,\xi\,}{\,2k_BT_c\,} \right) \right|
\leq \max_{\eta\geq 0}|g'(\eta)| \left(
 \frac{\,\hslash\omega_D\left| T-T_c \right|\,}{k_BT_c^2}+
 \frac{\,\sqrt{Y}\,}{\,k_BT_c\,} \right),
\]
it follows that
\[
\int_{2k_BT_c\,\varepsilon}^{\hslash\omega_D} \left\{
 g\left( \frac{\,\sqrt{\,\xi^2+Y\,}\,}{2k_BT} \right)
 -g\left( \frac{\,\xi\,}{\,2k_BT_c\,} \right) \right\}\,d\xi \to 0 \quad
 \mbox{as}\quad (T,\,Y) \to (T_c\,,\,0),
\]
and hence \quad $\displaystyle{
(\partial F/\partial Y)(T,\,Y)-(\partial F/\partial Y)(T_c\,,\,0) \to 0 \quad \mbox{as}\quad (T,\,Y) \to (T_c\,,\,0)}$.

\textit{Step 2}. When $(T,\,Y)=(T,\,0)\in W_3$ and $(T,\,Y)=(T_c\,,\,Y)\in W_4$, an argument similar to that in Step 1 gives
\[
\frac{\,\partial F\,}{\,\partial T\,}(T,\,0) \to
 \frac{\,\partial F\,}{\,\partial T\,}(T_c\,,\,0),\quad
\frac{\,\partial F\,}{\,\partial Y\,}(T,\,0) \to
\frac{\,\partial F\,}{\,\partial Y\,}(T_c\,,\,0) \quad \mbox{as}\quad
(T,\,0) \to (T_c\,,\,0)
\]
and
\[
\frac{\,\partial F\,}{\,\partial T\,}(T_c\,,\,Y) \to
 \frac{\,\partial F\,}{\,\partial T\,}(T_c\,,\,0),\quad
\frac{\,\partial F\,}{\,\partial Y\,}(T_c\,,\,Y) \to
\frac{\,\partial F\,}{\,\partial Y\,}(T_c\,,\,0)
\]
as \quad $(T_c\,,\,Y) \to (T_c\,,\,0)$. The result follows.
\end{proof}

\section{The second-order partial derivatives of the function $F$}

In this section we deal with the second-order partial derivatives of the function $F$ and show that $F$ is of class $C^2$ on $W_1$.

Let $G$ be given by
\begin{equation}\label{eq:fncg}
G(\eta)= \left\{ \begin{array}{ll}\displaystyle{
\frac{1}{\,\eta^2\,} \left\{ 3\,g(\eta)
+2\,\frac{\tanh\eta}{\,\eta\,\cosh^2\eta\,}\right\} }\qquad &(\eta>0),\\
\noalign{\vskip0.3cm} \displaystyle{
-\frac{\,16\,}{\,15\,} } &(\eta=0).
\end{array}\right.
\end{equation}

A straightforward calculation gives the following.
\begin{lemma}\label{lm:cgproperty}
Let $G$ be as in (\ref{eq:fncg}) and $g$ as in (\ref{eq:fng}). Then the
function $G$ is of class $C^1$ on $[0,\,\infty)$ and satisfies
\[
g'(\eta)=-\eta\, G(\eta),\qquad G'(0)=0,\qquad
\lim_{\eta\to\infty}G(\eta)=\lim_{\eta\to\infty}G'(\eta)=0.
\]
\end{lemma}

\begin{lemma}\label{lm:FTTexistence}
The values of the partial derivatives \  
$\displaystyle{ \frac{\,\partial^2 F\,}{\,\partial T^2\,} }$,\quad
$\displaystyle{ \frac{\partial}{\,\partial Y\,}\left( 
\frac{\,\partial F\,}{\,\partial T\,} \right) }$, \\
$\displaystyle{ \frac{\partial}{\,\partial T\,}\left( 
\frac{\,\partial F\,}{\,\partial Y\,} \right) }$ \  and \   
$\displaystyle{ \frac{\,\partial^2 F\,}{\,\partial Y^2\,} }$
exist at each point in $W_1$. Furthermore,
\[
\frac{\partial}{\,\partial Y\,}\left( \frac{\,\partial F\,}{\,\partial T\,}\right)=\frac{\partial}{\,\partial T\,}\left( 
\frac{\,\partial F\,}{\,\partial Y\,} \right) \qquad \mbox{on} \quad W_1\,.
\]
\end{lemma}

\begin{proof}
Let $(T,\,Y)\in W_1$. Then there is a $\theta$ \  \  $(0<\theta<1)$ satisfying $\theta T_c<T<T_c$. Set $\displaystyle{ \eta=\frac{\,\sqrt{\,\xi^2+Y\,}\,}{2k_BT} }$. Then
\[
\left|\frac{\partial}{\,\partial T\,} \frac{1}{\,\cosh^2\eta\,}\right|\leq \frac{\,\sqrt{\hslash^2\omega_D^2+2\,\Delta_0^2}\,}{k_B\theta^2T_c^2},
\]
where the right side is integrable on $[2k_BT_c\,\varepsilon,\,\hslash\omega_D]$. So the function: $(T,\,Y)\mapsto I_1(T,\,Y)$ (see (\ref{eq:I1I2})), and hence $(\partial F/\partial T)$ (see Lemma \ref{lm:FTFYW}) is differentiable with respect to $T$ on $W_1$, and the second-order partial derivative is given by
\[
\frac{\,\partial^2 F\,}{\,\partial T^2\,}(T,\,Y)
=\frac{1}{\,k_BT^3\,} \left\{ I_1(T,\,Y)
 -\int_{2k_BT_c\,\varepsilon}^{\hslash\omega_D}
 \frac{\eta\,\tanh\eta}{\,\cosh^2\eta \,}\,d\xi \right\},\qquad
 \eta=\frac{\,\sqrt{\,\xi^2+Y\,}\,}{2k_BT}\,.
\]

Similarly we can show that $(\partial F/\partial T)$ is differentiable with respect to $Y$ on $W_1$, that $(\partial F/\partial Y)$ is differentiable with respect to $T$ on $W_1$, and that $(\partial F/\partial Y)$ is differentiable with respect to $Y$ on $W_1$. The corresponding second-order partial derivatives are given as follows:
\[
\frac{\partial}{\,\partial Y\,}\left( 
\frac{\,\partial F\,}{\,\partial T\,} \right)(T,\,Y)
=\frac{\partial}{\,\partial T\,}\left( 
\frac{\,\partial F\,}{\,\partial Y\,} \right)(T,\,Y)
=\frac{1}{\,(2k_BT)^3T\,} \int_{2k_BT_c\,\varepsilon}^{\hslash\omega_D}
 \frac{\tanh\eta}{\,\eta\,\cosh^2\eta \,}\,d\xi,
\]
\begin{equation}\label{eq:fyy}
\frac{\,\partial^2 F\,}{\,\partial Y^2\,}(T,\,Y)
=-\,\frac{1}{\,4\,(2k_BT)^5\,}
\int_{2k_BT_c\,\varepsilon}^{\hslash\omega_D} G(\eta)\,d\xi\,,
\qquad \eta=\frac{\,\sqrt{\,\xi^2+Y\,}\,}{2k_BT}\,.
\end{equation}
Here, $G$ is that in Lemma \ref{lm:cgproperty} (see also (\ref{eq:fncg})).
\end{proof}

\begin{lemma}\label{lm:FC2W1}
The partial derivatives
$\displaystyle{ \frac{\,\partial^2 F\,}{\,\partial T^2\,} }$,
$\displaystyle{ \frac{\partial}{\,\partial Y\,}\left(
 \frac{\,\partial F\,}{\,\partial T\,} \right) }$ and
$\displaystyle{\frac{\,\partial^2 F\,}{\,\partial Y^2\,}}$
are continuous on $W_1$. Consequently, $F$ is of class $C^2$ on $W_1$.
\end{lemma}

\begin{proof}
We show that $(\partial^2 F/\partial Y^2)$ is continuous on $W_1$. Similarly we can show the continuity of other second-order partial derivatives.

By (\ref{eq:fyy}), it suffices to show that the function: $(T,\,Y)\mapsto I_3(T,\,Y)$ is continuous at $(T_0\,,\,Y_0)\in W_1$. Here,
\[
I_3(T,\,Y)=\int_{2k_BT_c\,\varepsilon}^{\hslash\omega_D}
 G(\eta)\,d\xi\,,\qquad \eta=\frac{\,\sqrt{\,\xi^2+Y\,}\,}{2k_BT}\,.
\]
Since $(T,\,Y)$ is close to $(T_0\,,\,Y_0)$, it then follows that
$T_0/2<T$. A straightforward calculation then gives
\begin{eqnarray}\nonumber
& &\left| \,I_3(T,\,Y)-I_3(T_0\,,\,Y_0)\,\right|\\ \nonumber
&\leq& \hslash\omega_D\,\displaystyle{\max_{\eta\geq 0}|G'(\eta)|}
\left( \frac{\,\sqrt{\,\hslash^2\omega_D^2+2\,\Delta_0^2\,}\,}{k_BT_0^2}|T-T_0|+\frac{|Y-Y_0|}{\,k_BT_0\sqrt{Y_0}\,}\right).
\end{eqnarray}
Hence the function: $(T,\,Y)\mapsto I_3(T,\,Y)$ is continuous at $(T_0\,,\,Y_0)\in W_1$.
\end{proof}

\section{Proof of Proposition \ref{prp:gapfunction}}

In this section we prove Proposition \ref{prp:gapfunction} in a sequence of lemmas.

\begin{remark} One may prove Proposition \ref{prp:gapfunction} on the basis of the implicit function theorem. In this case, \textit{an interior point} $(T_0,\,Y_0)$ of the domain $W$ satisfying $\displaystyle{ F(T_0,\,Y_0)=0 }$ need to exist. But there are the two points $(0,\,\Delta^2)$ and $(T_c\,,\,0)$ in \textit{the boundary} of $W$ satisfying
\begin{equation}\label{eq:F=0}
F(0,\,\Delta^2)=F(T_c\,,\,0)=0.
\end{equation}
So one can not apply the implicit function theorem in its present form.
\end{remark}

\begin{lemma}\label{lm:existence}
There is a unique solution: $T \mapsto Y=f(T)$ to the gap equation $\displaystyle{F(T,\,Y)=0}$ such that the function $f$ is continuous on the closed interval $[0,\,T_c]$ and satisfies $f(0)=\Delta^2$ and $f(T_c)=0$.
\end{lemma}

\begin{proof}
By Lemmas \ref{lm:FTFYminus}, \ref{lm:FC1onW} and (\ref{eq:F=0}), the function: $Y \mapsto F(T_c\,,\,Y)$ is monotonically decreasing and there is a $Y_1$ \   $(0<Y_1<2\Delta_0^2)$ satisfying $F(T_c\,,\,Y_1)<0$. Note that $Y_1$ is arbitrary as long as $0<Y_1<2\Delta_0^2$. Hence, by Lemma \ref{lm:FC1onW}, there is a $T_1$ \   $(0<T_1<T_c)$ satisfying $F(T_1\,,\,Y_1)<0$. Hence, $F(T,\,Y_1)<0$ for $T_1\leq T\leq T_c$. On the other hand, by Lemmas \ref{lm:FTFYminus}, \ref{lm:FC1onW} and (\ref{eq:F=0}), the function: $T \mapsto F(T,\,0)$ is monotonically decreasing and there is a $T_2$ \   $(0<T_2<T_c)$ satisfying $F(T_2\,,\,0)>0$. Note that $T_2$ is arbitrary as long as $0<T_2<T_c$. Hence, $F(T,\,0)>0$ for $T_2\leq T<T_c$.

Let $\max(T_1\,,\,T_2)\leq T<T_c$ and fix $T$. It then follows from Lemmas \ref{lm:FTFYminus} and \ref{lm:FC1onW} that the function: $Y \mapsto F(T,\,Y)$ with $T$ fixed is monotonically decreasing on $[0,\,Y_1]$. Since $F(T,\,0)>0$ and $F(T,\,Y_1)<0$, there is a unique $Y$ \   $(0<Y<Y_1)$ satisfying $F(T,\,Y)=0$. When $T=T_c$, there is a unique value $Y=0$ satisfying $F(T_c\,,\,Y)=0$ \  (see (\ref{eq:F=0})).

Since $F$ is continuous on $W$ by Lemma \ref{lm:FC1onW}, there is a unique solution: $T \mapsto Y=f(T)$ to the gap equation $\displaystyle{F(T,\,Y)=0}$ such that the function $f$ is continuous on $[\max(T_1\,,\,T_2),\,T_c]$ and $f(T_c)=0$.

Since $(\partial F/\partial Y)(0,\,Y)<0$ \   $(0<Y<2\Delta_0^2)$ by Lemma
\ref{lm:FTFYW}, there is a unique value $Y=\Delta^2$ satisfying
$F(0,\,Y)=0$. Combining Lemma \ref{lm:FC1onW} with Lemma \ref{lm:FTFYminus} therefore implies that the function $f$ is continuous on $[0,\,T_c]$ and that $f(0)=\Delta^2$ and $f(T_c)=0$.
\end{proof}

\begin{lemma}\label{lm:derivative}
The function $f$ given by Lemma \ref{lm:existence} is of class $C^1$ on $[0,\,T_c]$, and the derivative $f'$ satisfies
\[
f'(0)=0,\qquad f'(T_c)=8\, k_B^2T_c\,\frac{\,\displaystyle{
 \int_{\displaystyle{\varepsilon}}
 ^{\displaystyle{\hslash\omega_D/(2k_BT_c)}}
 \frac{d\eta}{\,\cosh^2\eta\,} }\,}
{\,\displaystyle{ 
 \int_{\displaystyle{\varepsilon}}
 ^{\displaystyle{\hslash\omega_D/(2k_BT_c)}}
 g(\eta)\,d\eta }\,}\,.
\]
\end{lemma}

\begin{proof}
Lemma \ref{lm:FC1onW} immediately implies that the function $f$ is of class $C^1$ on the interval $[0,\,T_c]$ and that its derivative is given by
\begin{equation}\label{eq:f'(T)}
f'(T)=-\frac{\,F_T(T,\,f(T))\,}{\,F_Y(T,\,f(T))\,}\,.
\end{equation}
The values of $f'(0)$ and $f'(T_c)$ are derived from (\ref{eq:f'(T)}).
\end{proof}

Combining (\ref{eq:f'(T)}) with Lemma \ref{lm:FTFYminus} immediately yields the following.
\begin{lemma}
The function $f$ given by Lemma \ref{lm:existence} is monotonically decreasing on $[0,\,T_c]$:
\[
f(0)=\Delta^2>f(T_1)>f(T_2)>f(T_c)=0, \qquad 0<T_1<T_2<T_c\,.
\]
\end{lemma}

Let $\phi$ be a function of $\eta$ and let $\eta$ be a function of $\xi$.
Set
\begin{equation}\label{eq:integralI}
I\left[ \,\phi(\eta) \,\right]=
\int_{2k_BT_c\,\varepsilon}^{\hslash\omega_D}\phi(\eta)\,d\xi\,.
\end{equation}

\begin{lemma}\label{lm:classc2}
Let $f$ be given by Lemma \ref{lm:existence} and let $I\left[ \cdot \right]$ be as in (\ref{eq:integralI}). Then the function $f$ is of class $C^2$ on $[0,\,T_c]$, and the second derivative $f''$ satisfies \quad $\displaystyle{ f''(0)=0 }$
\quad and
\begin{eqnarray}\nonumber
& & f''(T_c) \\ \nonumber
&=& 16\, k_B^2\,\frac{\,\displaystyle{
 I\left[ \,\frac{\,\eta_0\tanh \eta_0-1\,}{\cosh^2 \eta_0}\, \right]
}\,}{\, \displaystyle{
 I\left[ \,g(\eta_0)\, \right]  }\,}
-32\, k_B^2\,\frac{\,\displaystyle{
 I\left[ \,\frac{1}{\,\cosh^2 \eta_0\,}\, \right]
 I\left[ \,\frac{\,\tanh \eta_0\,}{\,\eta_0\cosh^2 \eta_0\,}\, \right]
 }\,}{\, \displaystyle{
 \left\{ \, I\left[ \,g(\eta_0)\, \right] \, \right\}^2  }\,} \\ \nonumber
& &\quad +8\, k_B^2\,\frac{\,\displaystyle{
 \left\{ I\left[ \,\frac{1}{\,\cosh^2 \eta_0\,}\, \right] \right\}^2
 I\left[ \,G(\eta_0)\, \right]
 }\,}{\, \displaystyle{
 \left\{ \, I\left[ \,g(\eta_0)\, \right] \, \right\}^3 }\,}\,,
\qquad \eta_0=\frac{\xi}{\,2k_BT_c\,}\,.
\end{eqnarray}
\end{lemma}

\begin{proof}
Lemma \ref{lm:FC2W1} implies that $f$ is of class $C^2$ on the open interval $(0,\,T_c)$ and that
\begin{equation}\label{eq:f''}
f''(T)=\frac{\,-F_{TT}F_Y^2+2F_{TY}F_TF_Y-F_{YY}F_T^2\,}{F_Y^3}\,,\qquad
0<T<T_c\,.
\end{equation}
So we have only to deal with $f$ and its derivatives at $T=0$ and at $T=T_c$.

\textit{Step 1}. We show that $f'$ is differentiable at $T=0$ and that
$f''$ is continuous at $T=0$.

Note that $f'(0)=0$ by Lemma \ref{lm:derivative}. Since $T$ is close to $T=0$, the inequality $f(T)>\Delta_0^2\, /2$ holds. It then follows from (\ref{eq:f'(T)}) and Lemma \ref{lm:FTFYW} that
\[
\left| \frac{\,f'(T)-f'(0)\,}{T} \right| \leq
4\,\frac{ \,\displaystyle{
 {\;\sqrt{\hslash^2\omega_D^2+2\,\Delta_0^2}\;}^3
 \exp\left(-\frac{\,\sqrt{\Delta_0^2/2}\,}{k_BT}\right) }
 \,}{ \displaystyle{ 
 k_BT^3\left( \tanh \eta_1-\frac{\eta_1}{\,\cosh^2\eta_1\,}\right) }
 } \to 0 \qquad (T\downarrow 0).
\]
Here, \  $\displaystyle{
\eta_1=\frac{\,\sqrt{\,\xi_1^2+f(T)\,}\,}{2k_BT} \to \infty }$ \  
as $T\downarrow 0$ \quad $(2k_BT_c\,\varepsilon<\xi_1<\hslash\omega_D)$.
Hence $f'$ is differentiable at $T=0$ and $f''(0)=0$.

By (\ref{eq:f''}), a similar argument gives \quad $\displaystyle{ 
\lim_{T\downarrow 0} f''(T)=0 }$. Hence $f''$ is continuous at $T=0$.

\textit{Step 2}. We show that $f'$ is differentiable at $T=T_c$ and that $f''$ is continuous at $T=T_c$.

Note that
\[
f'(T_c)=8\, k_B^2T_c\,\frac{\,\displaystyle{
 I\left[ \,\frac{1}{\,\cosh^2 \eta_0\,}\, \right]
}\,}{\,\displaystyle{ 
 I\left[ \,g(\eta_0)\, \right]
}\,}\,,
\qquad \eta_0=\frac{\xi}{\,2k_BT_c\,}
\]
by Lemma \ref{lm:derivative}. It follows from (\ref{eq:f'(T)}) and Lemma
\ref{lm:FTFYW} that
\[
f'(T)=8\, k_B^2T \,\frac{\,\displaystyle{
 I\left[ \,\frac{1}{\,\cosh^2 \eta\,}\, \right]
}\,}{\,\displaystyle{ I\left[ \,g(\eta)\, \right]
}\,}\,, \qquad \eta=\frac{\,\sqrt{\,\xi^2+f(T)\,}\,}{2k_BT}\,.
\]
Hence
\begin{eqnarray}\nonumber
& & \frac{\,f'(T_c)-f'(T)\,}{T_c-T} \\ \nonumber
&=&8\, k_B^2\,\frac{\,\displaystyle{
 I\left[ \,\frac{1}{\,\cosh^2 \eta_0\,}\, \right]  }\,}{\,\displaystyle{ 
 I\left[ \,g(\eta_0)\, \right]  }\,}+\frac{8\, k_B^2T}{\,T_c-T\,}
\frac{
 I\left[ \,\frac{1}{\,\cosh^2 \eta_0\,}\, \right] \left\{
 I\left[ \,g(\eta)\, \right]-I\left[ \,g(\eta_0)\, \right] \right\}
}{
I\left[ \,g(\eta_0)\, \right] I\left[ \,g(\eta)\, \right]
} \\ \nonumber
& & \quad +\frac{8\, k_B^2T}{\,T_c-T\,}\,
\frac{\,
 I\left[ \,g(\eta_0)\, \right] \left\{
 I\left[ \,\frac{1}{\,\cosh^2 \eta_0\,}\, \right]
 -I\left[ \,\frac{1}{\,\cosh^2 \eta\,}\, \right] \right\}
\,}{
I\left[ \,g(\eta_0)\, \right] \, I\left[ \,g(\eta)\, \right]
}\,.
\end{eqnarray}
Note that \quad $g(\eta)-g(\eta_0)=(\eta-\eta_0)g'(\eta_1)$ and
$\cosh\eta-\cosh\eta_0=(\eta-\eta_0)\sinh\eta_2$. Here,
\[
\eta_0=\frac{\xi}{\,2k_BT_c\,}<\eta_i<\eta=\frac{\,\sqrt{\,\xi^2+f(T)\,}\,}{2k_BT}\,,\qquad i=1,\,2
\]
and
\[
\eta-\eta_0=\frac{1}{\,2k_BT\,}\left\{
\frac{f(T)}{\, \sqrt{\,\xi^2+f(T)\,}+\xi \,}+\xi\frac{\,T_c-T\,}{T_c}
\right\}.
\]
Since $T$ is close to $T_c$, the inequality $T>T_c\, /2$ holds. Therefore,
by Lemma \ref{lm:cgproperty},
\[
\left| \frac{g'(\eta_1)}{\,\sqrt{\,\xi^2+f(T)\,}+\xi\,} \right| \leq
\frac{1}{\,k_BT_c\,}\,\max_{\eta\geq 0}\left| G(\eta) \right|
\]
and
\[
\left| \frac{\sinh\eta_2}{\,\sqrt{\,\xi^2+f(T)\,}+\xi\,} \right| \leq
\frac{1}{\,k_BT_c\,}\,\max_{0 \leq\eta\leq M}
 \left| \frac{\,\sinh\eta\,}{\eta} \right|,\qquad
M=\frac{\,\sqrt{ \hslash^2\omega_D^2+2\,\Delta_0^2 }\,}{k_BT_c}\,.
\]
So $f'$ is differentiable at $T=T_c$, and it is easy to see that the form of $f''(T_c)$ is exactly the same as that mentioned just above.

Furthermore, it follows from (\ref{eq:f''}) that $f''$ is continuous at
$T=T_c$.
\end{proof}

\section{Proof of Theorem \ref{thm:phasetransition}}

In this section we prove Theorem \ref{thm:phasetransition} in a sequence of lemmas. Fixing the values of the chemical potential $\mu$ and of
the volume of our physical system, we deal with the dependence of the thermodynamical potential $\Omega$ on the temperature $T$ only.

\begin{lemma}\label{lm:v}
Let $V$ be as in \eqref{eq:v}. Then $V$ is of class $C^2$ on
$(0,\,\infty)$.
\end{lemma}

\begin{proof}
For each $T>0$, there are a $\theta_1$ \   $(0<\theta_1<1)$ and a $\theta_2$ \   $(\theta_2>1)$ satisfying $\theta_1T_c<T<\theta_2T_c$. Then
\[
\left| \frac{\partial}{\,\partial T\,} \ln
 \left( 1+e^{\displaystyle{-|\xi|/(k_BT)}} \right) \right|\leq
\frac{\, |\xi| \, e^{\displaystyle{-|\xi|/(k_B\theta_2T_c)}}\,}
{k_B\theta_1^2T_c^2},
\]
where the right side is integrable on $[-\mu,\, -\hslash\omega_D]$ and on $[\hslash\omega_D,\, \infty)$ since $N(\xi)=O(\sqrt{\xi})$ as $\xi \to \infty$ (see Remark \ref{rmk:nxi}). Hence $V$ is differentiable on $(0,\,\infty)$ and
\begin{eqnarray}\nonumber
\frac{\,\partial V\,}{\partial T}(T)&=&-2k_B
\int_{[-\mu,\,-\hslash\omega_D]\,\cup\,[\hslash\omega_D,\,\infty)}
 N(\xi)\ln\left( 1+e^{\displaystyle{-|\xi|/(k_BT)}} \right)\,d\xi
 \\  \nonumber
& &\quad -\frac{2}{\, T\,}
\int_{[-\mu,\,-\hslash\omega_D]\,\cup\,[\hslash\omega_D,\,\infty)}
 N(\xi)\,\frac{|\xi|}{\,1+e^{\displaystyle{|\xi|/(k_BT)}}\,}\,d\xi.
\end{eqnarray}

A similar argument gives
\[
\left| \frac{\partial}{\,\partial T\,}\,
 \frac{1}{\,1+e^{\displaystyle{|\xi|/(k_BT)}}\,} \right|\leq
\frac{\, |\xi| \, e^{\displaystyle{-|\xi|/(k_B\theta_2T_c)}}\,}
{k_B\theta_1^2T_c^2},
\]
where the right side is integrable on $[-\mu,\, -\hslash\omega_D]$ and on $[\hslash\omega_D,\, \infty)$. Therefore, $(\partial V/\partial T)$
is again differentiable on $(0,\,\infty)$ and
\[
\frac{\,\partial^2 V\,}{\,\partial T^2\,}(T)=-\frac{2}{\,k_BT^3\,}
\int_{[-\mu,\,-\hslash\omega_D]\,\cup\,[\hslash\omega_D,\,\infty)}
 N(\xi)\,\frac{|\xi|^2\, e^{\displaystyle{|\xi|/(k_BT)}}}
 {\,\left( 1+e^{\displaystyle{|\xi|/(k_BT)}} \right)^2\,}\,d\xi.
\]
Clearly, $(\partial^2 V/\partial T^2)$ is continuous on $(0,\,\infty)$.
\end{proof}

\begin{lemma}\label{lm:OmegaN}
Let $\Omega_N$ be as in \eqref{eq:omegan}. Then $\Omega_N$ is of
class $C^2$ on $(0,\,\infty)$.
\end{lemma}

\begin{proof}
An argument similar to that in the proof of Lemma \ref{lm:v} gives that
$\Omega_N$ is of class $C^2$ on $(0,\,\infty)$ and that the derivatives
are given by
\begin{eqnarray}\nonumber
\frac{\,\partial \Omega_N\,}{\,\partial T\,}(T)&=&-4N_0k_B
\int_{\displaystyle{2k_BT_c\,\varepsilon}}^{\displaystyle{\hslash\omega_D}} \ln\left( 1+e^{\displaystyle{-\xi/(k_BT)}} \right)\,d\xi \\ \nonumber
& &\quad -\frac{\, 4N_0\,}{T}
\int_{\displaystyle{2k_BT_c\,\varepsilon}}^{\displaystyle{\hslash\omega_D}} \frac{\xi}{\,1+e^{\displaystyle{\xi/(k_BT)}}\,}\,d\xi
 +\frac{\,\partial V\,}{\partial T}(T), \\ \nonumber
\frac{\,\partial^2 \Omega_N\,}{\,\partial T^2\,}(T)
&=&-\frac{\, 4N_0\,}{\,k_BT^3\,}
\int_{\displaystyle{2k_BT_c\,\varepsilon}}^{\displaystyle{\hslash\omega_D}}\frac{\xi^2\,e^{\displaystyle{\xi/(k_BT)}}}
{\,\left( 1+e^{\displaystyle{\xi/(k_BT)}}\right)^2\,}\,d\xi
+\frac{\,\partial^2 V\,}{\,\partial T^2\,}(T).
\end{eqnarray}
\end{proof}

\begin{lemma}\label{lm:delta}
Let $\delta$ be as in \eqref{eq:delta}. Then $\delta$ is of
class $C^2$ on $(0,\,T_c]$.
\end{lemma}

\begin{proof}
Note that the squared gap function $f$ is of class $C^2$ on $[0,\,T_c]$ by Lemma \ref{lm:classc2} and that
\begin{equation}\label{eq:deltaTc}
\delta(T_c)=0
\end{equation}
since $f(T_c)=0$ (see Lemma \ref{lm:existence}). A straightforward calculation gives that $\delta$ is continuous on $(0,\,T_c]$ and that
\[
\left| \frac{\partial}{\,\partial T\,} \sqrt{\xi^2+f(T)} \right| \leq
\frac{\,\displaystyle{\max_{0 \leq T \leq T_c} \left| f'(T) \right|}\,}
 {\,4k_BT_c\,\varepsilon\,},
\]
where the right side is integrable on $[2k_BT_c\,\varepsilon,\, \hslash\omega_D]$. By an argument similar to that in the proof of Lemma \ref{lm:v}, $\delta$ is differentiable on $(0,\,T_c]$ and the derivative is given by
\begin{eqnarray}\nonumber
\frac{\,\partial \delta\,}{\,\partial T\,}(T)&=&f'(T)
\left\{ \frac{1}{\,U_0\,}-N_0
\int_{\displaystyle{2k_BT_c\,\varepsilon}}^{\displaystyle{\hslash\omega_D}}\frac{1}{\,\sqrt{\,\xi^2+f(T)\,}\,}
\tanh\frac{\, \sqrt{\,\xi^2+f(T)\,}\,}{2k_BT}\,d\xi \right\}\\ \nonumber
& &\quad -4N_0k_B
\int_{\displaystyle{2k_BT_c\,\varepsilon}}^{\displaystyle{\hslash\omega_D}}\ln \frac{\,1+e^{-\displaystyle{\sqrt{\xi^2+f(T)}/(k_BT)}}\,}
{1+e^{-\displaystyle{\xi/(k_BT)}}} \,d\xi \\ \nonumber
& &\quad +\frac{\, 4N_0\,}{T}
\int_{\displaystyle{2k_BT_c\,\varepsilon}}^{\displaystyle{\hslash\omega_D}}\left\{ \frac{\xi}{\,1+e^{\displaystyle{\xi/(k_BT)}}\,}-
\frac{\sqrt{\xi^2+f(T)}}
{\,1+e^{\displaystyle{\sqrt{\xi^2+f(T)}/(k_BT)}}\,} \right\}\,d\xi,
\end{eqnarray}
where the first term on the right side is equal to 0 by the gap equation (\ref{eq:gapequation}).
Note that
\begin{equation}\label{eq:delta'Tc}
\frac{\,\partial \delta\,}{\,\partial T\,}(T_c)=0.
\end{equation}

An argument similar to that in the proof of Lemma \ref{lm:v} gives that
$(\partial \delta/\partial T)$ is again differentiable on $(0,\,T_c]$ and
the second-order derivative is given by
\begin{eqnarray}\nonumber
& &\frac{\,\partial^2 \delta\,}{\,\partial T^2\,}(T) \\ \nonumber
&=&\frac{4N_0}{\,k_BT^3\,}
\int_{\displaystyle{2k_BT_c\,\varepsilon}}^{\displaystyle{\hslash\omega_D}}\frac{\xi^2\, e^{\displaystyle{\xi/(k_BT)}}}
{\,\left( 1+e^{\displaystyle{\xi/(k_BT)}} \right)^2\,}\,d\xi \\ \nonumber
& &-\frac{4N_0}{\,k_BT^3\,}
\int_{\displaystyle{2k_BT_c\,\varepsilon}}^{\displaystyle{\hslash\omega_D}}\frac{ e^{\displaystyle{\sqrt{\xi^2+f(T)}/(k_BT)}} }
{\,\left( 1+e^{\displaystyle{\sqrt{\xi^2+f(T)}/(k_BT)}}\right)^2\,}
\left\{ \xi^2+f(T)-\frac{\,Tf'(T)\,}{2} \right\} \,d\xi,
\end{eqnarray}
which is also continuous on $(0,\,T_c]$. Thus $\delta$ is of class $C^2$ on $(0,\,T_c]$, and
\begin{equation}\label{eq:delta''Tc}
\frac{\,\partial^2 \delta\,}{\,\partial T^2\,}(T_c)
=\frac{\,2N_0f'(T_c)\,}{T_c} \left(
\frac{1}{\,1+e^{\displaystyle{2\varepsilon}} \,}-
\frac{1}{\,1+e^{\displaystyle{\hslash\omega_D/(k_BT_c)}} \,} \right).
\end{equation}
\end{proof}

We now give a proof of Theorem \ref{thm:phasetransition}.
\begin{lemma}
Let $f'(T_c)$ be given by Lemma \ref{lm:derivative} and let $\Omega$ be the thermodynamical potential given by Definition \ref{dfn:thpo}.

{\rm (i)}\quad The thermodynamical potential $\Omega$, regarded as
a function of $T$, is of class $C^1$ on $(0,\,\infty)$.

{\rm (ii)}\quad The second-order derivative
$\left( \partial^2\Omega/\partial T^2\right)$ is continuous on
$(0,\,\infty) \setminus \{ T_c\}$.

{\rm (iii)} \[
\lim_{T\uparrow T_c}\frac{\,\partial^2\Omega\,}{\,\partial T^2\,}(T)-
\lim_{T\downarrow T_c} \frac{\,\partial^2\Omega\,}{\,\partial T^2\,}(T)
=\frac{\,2N_0f'(T_c)\,}{T_c} \left(
\frac{1}{\,1+e^{\displaystyle{2\varepsilon}} \,}-
\frac{1}{\,1+e^{\displaystyle{\hslash\omega_D/(k_BT_c)}} \,} \right).
\]
\end{lemma}

\begin{proof}
Note that \  $\displaystyle{
\delta(T_c)=(\partial \delta/\partial T)(T_c)=0 }$ \  
(see (\ref{eq:deltaTc}) and (\ref{eq:delta'Tc})). Hence both (i) and (ii)
follow immediately from Lemmas \ref{lm:v}, \ref{lm:OmegaN} and
\ref{lm:delta}. Since
\[
\lim_{T\uparrow T_c}\left( \partial^2\Omega/\partial T^2 \right)(T)-
\lim_{T\downarrow T_c}\left( \partial^2\Omega/\partial T^2 \right)(T)
=(\partial^2 \delta/\partial T^2)(T_c),
\]
(iii) follows immediately from (\ref{eq:delta''Tc}).
\end{proof}

\begin{remark}
This lemma implies that the transition to a superconducting state at the transition temperature $T_c$ is a second-order phase transition.
\end{remark}


\end{document}